\documentclass[preprint,12pt]{elsarticle}

\usepackage{graphicx,url}
\usepackage[utf8]{inputenc}  
\usepackage{amsmath}
\usepackage{amssymb}
\usepackage{amsthm}
\usepackage{latexsym}
\usepackage{float}

\allowdisplaybreaks

\newcommand{\Nothing}{{\tt fail}}
\newcommand{\Just}[1]{#1}
\newcommand{\fivespaces}{\;\;\;\;\;}
\newcommand{\tenspaces}{\fivespaces\fivespaces}

\newcommand{\interf}{\fivespaces}
\newcommand{\mylabel}[1]{\, \mathbf{(#1)}}
\newcommand{\Con}[2]{#1\,#2}
\newcommand{\Choice}[2]{#1\,|\:#2}
\newcommand{\Chre}[2]{#1\,|\,#2}
\newcommand{\Choexe}[2]{#1\;\;|\;\;#2}

\newcommand{\Gk}{G_k}

\newcommand{\Wp}{w^\prime}
\newcommand{\Anyx}{X}
\newcommand{\Bvert}{\big\vert}
\newcommand{\Lp}{\stackrel{\mbox{\tiny{P\!E\!G}}}{\leadsto}}
\newcommand{\Nlp}{\not \stackrel{\mbox{\tiny{P\!E\!G}}}{\leadsto}}
\newcommand{\Lr}{\stackrel{\mbox{\tiny{R\!E}}}{\leadsto}}
\newcommand{\Epsi}{\varepsilon}
\newcommand{\Rp}[2]{\Pi(#1,\,#2)}
\newcommand{\Tup}[2]{(#1,\,#2)}
\newcommand{\Peg}[2]{#1[#2]}
\newcommand{\Pgg}[1]{\Peg{G}{#1}}
\newcommand{\Pgk}[1]{\Peg{\Gk}{#1}}
\newcommand{\Grm}[4]{(#1,\,#2,\,#3,\,#4)}
\newcommand{\Mat}[2]{#1\;\,#2\,}
\newcommand{\Reg}[2]{#1\;\,#2\,}
\newcommand{\Matg}[2]{\Mat{\Pgg{#1}}{#2}}
\newcommand{\Matk}[2]{\Mat{\Pgk{#1}}{#2}}

\newcommand{\Pk}{p_k}
\newcommand{\Vk}{V_k}
\newcommand{\Ek}{e_k}

\newcommand{\Ch}[1]{\texttt{#1}}

\newcommand{\Fst}{FIRST}

\newcommand{\arrow}{\rightarrow}

\newcommand{\Pdot}{\textbf{.}}

\newcommand{\Fout}{f_{\mathit{out}}}
\newcommand{\Fin}{f_{\mathit{in}}}
\newcommand{\Isnull}{\mathit{empty}}
\newcommand{\Nnull}{\neg\Isnull}
\newcommand{\Hase}{\mathit{null}}
\newcommand{\Nemp}{\neg\Hase}
\newcommand{\True}{\textbf{true}}
\newcommand{\False}{\textbf{false}}
\newcommand{\Lor}[2]{#1 \;\vee\; #2}
\newcommand{\Land}[2]{#1 \;\wedge\; #2}
\newcommand{\E}{e}

\newcommand{\Eind}{?\rangle e_1}
\newcommand{\Eindd}[1]{?\rangle #1}
\newcommand{\Epos}{e_1^{*+}}
\newcommand{\Eposs}[1]{#1^{*+}}
\newcommand{\Elazy}{e_1^{*?}}
\newcommand{\Elazyy}[1]{#1^{*?}}
\newcommand{\Enot}{?!e_1}
\newcommand{\Enott}[1]{?!#1}
\newcommand{\Eand}{?\!\!=\!e_1}
\newcommand{\Eandd}[1]{?\!\!=\!#1}

\newcommand{\Myeq}{\;=\;}

\newtheorem{lemma}{Lemma}

\journal{Science of Computer Programming}

\begin{document}

\begin{frontmatter}

\title{From Regexes to Parsing Expression Grammars}

\author{S\'{e}rgio Medeiros}
\ead{sergio@ufs.br}
\address{Department of Computer Science -- UFS -- Aracaju -- Brazil} 

\author{Fabio Mascarenhas}
\ead{mascarenhas@ufrj.br}
\address{Department of Computer Science -- UFRJ -- Rio de Janeiro -- Brazil} 

\author{Roberto Ierusalimschy}
\ead{roberto@inf.puc-rio.br}
\address{Department of Computer Science -- PUC-Rio -- Rio de Janeiro -- Brazil} 

\begin{abstract}
Most scripting languages nowadays use
regex pattern-matching libraries.
These regex libraries borrow the syntax of regular expressions, but
have an informal semantics that is different from the semantics of
regular expressions, removing the commutativity of alternation and
adding ad-hoc extensions that cannot be expressed by formalisms for
efficient recognition of regular languages, such as deterministic
finite automata.

Parsing Expression Grammars are a formalism that can describe all
deterministic context-free languages and has a simple computational model.
In this paper, we present a  formalization of
regexes via transformation to Parsing Expression Grammars.
The proposed transformation easily accommodates
several of the common regex extensions, giving a formal meaning to them.
It also provides a clear computational model that
helps to estimate the efficiency of regex-based matchers, and a basis for
specifying provably correct optimizations for them.
\end{abstract}


\end{frontmatter}

\section{Introduction}
Regular expressions are a concise way for describing regular languages
with an algebraic notation. Their syntax has seen wide use in pattern
matching libraries for programming languages, where they are used to
specify the pattern against which a user is trying to match a string,
or, more commonly, the pattern that a user is searching for in a
string. 

Regular expressions used for pattern matching are known as
{\em regexes}~\cite{apoc5,synops5}, and while they look like regular
expressions they often have different semantics, based on how the
pattern matching libraries that use them are actually implemented.

A simple example that shows this semantic difference are the regular expressions
$a | aa$ and $aa | a$, which both describe the language $\{ \Ch{a}, \Ch{aa}
\}$. It is trivial to prove that the $|$ operator of regular
expressions is commutative given its common semantics as the union of
sets. But the {\em regexes} $a | aa$ and $aa | a$ behave
differently for some implementations of pattern matching libraries and
some subjects.

The standard regex libraries of the Perl and Ruby languages,
as well as PCRE~\cite{pcre}, a regex library with bindings for many
programming languages, give different results when matching these two
regexes against the subject {\tt aa}. In all three libraries the
first regex matches just the first {\tt a} of the subject, while the
second regex matches the whole subject. With the subject {\tt ab} both regexes
give the same answer in all three libraries, matching the first {\tt
  a}, but we can see that the $|$ operator for regexes is not
commutative.

This behavior of regexes is directly linked to the way they are usually
implemented, by trying the alternatives in a $|$ expression in the
order they appear and backtracking when a particular path
through the expression makes the match fail. 

A naive implementation of regex matching via backtracking can have exponential worst-case running time, which
implementations try to avoid through ad-hoc optimizations to cut the
amount of backtracking that needs to be done for common patterns. These ad-hoc
optimizations lead to implementations not having a cost model of their
operation, which makes it difficult for users to determine the performance of regex
patterns. Simple modifications can make the time complexity of a pattern
go from linear to exponential in unpredictable ways~\cite{rewbr,nphard}.

Regexes can also have syntactical and semantical extensions that are
difficult, or even impossible, to express through pure regular
expressions. These extensions do not have a formal model, but are
informally specified through how they modify the behavior of an
implementation based on backtracking. The meaning of regex patterns
that use the extensions may vary among different regex libraries~\cite{friedl:regex}, or even
among different implementations of the same regex library~\cite{fowler:posix}.

Practical regex libraries try to solve performance problems with
ad-hoc optimizations for common patterns, but this makes the
implementation of a regex library a complex task, and is another
source of unpredictable performance, as different implementations can
have different performance characteristics.

A heavily optimized regex engine, RE2~\cite{regwild}, uses an implementation based on
finite automata and guarantees linear time performance, but it relies
on ad-hoc optimizations to handle more complex patterns, as a naive
automata-based implementation can have quadratic
behavior~\cite{maxmunch}. More importantly, it cannot implement some
 common regex extensions~\cite{regwild}.

Parsing Expression Grammars (PEGs)~\cite{ford:peg} are a formalism
that can express all deterministic context-free languages, which means
that PEGs can also express all regular languages. The syntax of PEGs
is based on the syntax of regular expressions and regexes, and PEGs have a
formal semantics based on {\em ordered choice}, a controlled form of
backtracking that, like the $|$ operation of
regexes, is sensitive to the ordering of the alternatives. 

We believe that ordered choice makes PEGs a suitable base for a formal
treatment of regexes, and show, in this paper, that we can describe
the meaning of regex patterns by conversion to PEGs. Moreover, PEGs
can be efficiently executed by a parsing machine that has a clear cost
model that we can use to reason
about the time complexity of matching a given
pattern~\cite{roberto:lpeg,dls:lpeg}. We can then use
the semantics of PEGs to reason about the behavior of regexes, for
example, to optimize pattern matching and searching by avoiding
excessive backtracking. We believe that the combination of the
regex to PEG conversion and the PEG parsing machine can be used to build
implementations of regex libraries that are simpler and easier to
extend than current ones.

The main contribution of this paper is our formalization of a simple,
structure-preserving translation from plain regular expressions to
PEGs that can be used to translate regexes to PEGs that match the same
subjects. We present a formalization of regular expressions as
patterns that match prefixes of strings instead of sets of strings,
using the framework of natural semantics~\cite{kahn:semantics}. In
this semantics, regular expressions are just a non-deterministic form
of the regexes used for pattern matching. We show that our semantics
is equivalent to the standard set-based semantics when we consider the
language of a pattern as the set of prefixes that it matches.

We then present a formalization of PEGs in the same style, and use
it to show the similarities and differences between regular
expressions, regexes, and PEGs. We then define a transformation that converts a
regular expression to a PEG, and prove its correctness. We also show
how we can improve the transformation for some classes of regexes by
exploiting their properties and the greater predictability and control
of performance that PEGs have, improving
the performance of the resulting PEGs. Finally, we show how our
transformation can be adapted to accommodate four regex extensions
that cannot be expressed by regular expressions:
independent expressions, possessive and lazy repetition, and lookahead.

There are procedures for transforming deterministic finite automata
and right-linear grammars to PEGs~\cite{roberto:lpeg,sergio:tese} and,
as there are transformations from regular expressions to these
formalisms, we could have used these existing procedures as the basis
of an implementation of regular expressions in PEG engines. But the
transformations of regular expressions to these formalisms cover just
a subset of regexes, not including common extensions, including those
covered in Section~\ref{sec:exts} of this paper. The direct
transformation we present here is straightforward and can cover regex extensions.

In the next section, we present our formalizations of regular expressions
and PEGs, and discuss when a regular expression has the same meaning
when interpreted as a regular expression and as a PEG, along with the
intuition behind our transformation. In Section~\ref{sec:equiv}, we
formalize our transformation from regular expressions to PEGs and
prove its correctness. In Section~\ref{sec:opt} we show how we can
reason about the performance of PEGs to improve the PEGs generated by
our transformation in some cases. In Section~\ref{sec:bench} we show
how our approach compares to existing regex implementations with some
benchmarks. In Section~\ref{sec:exts} we show
how our transformation can accommodate some regex extensions. Finally,
in Section~\ref{sec:conclusion} we discuss some related work and present our conclusions.

\section{Regular Expressions and PEGs}
\label{sec:re}

Given a finite alphabet $T$, we can define a regular expression
$\E$ inductively as follows, where $\Ch{a} \in T$, and both $e_1$ and $e_2$
are also regular expressions:
\begin{align*}
& \E \;=\; \Epsi \;\;\Bvert\;\;
          a \;\;\Bvert\;\; \Con{e_1}{e_2} \;\;\Bvert\;\;
          \Choice{e_1}{e_2} \;\;\Bvert\;\; e_1^*
\end{align*}

Traditionally, a regular expression can also be $\emptyset$, but we
will not consider it; $\emptyset$ is not used in regexes, and any
expression with $\emptyset$ as a subexpression can be either rewritten
without $\emptyset$ or is equal to $\emptyset$. 

Note that this definition  gives an abstract syntax for
expressions, and this abstract syntax is what we use in the formal
semantics and proofs. In our examples we use a concrete syntax that
assumes that juxtaposition (concatenation) is left-associative and has
higher precedence than $|$, which is also left-associative, while $*$ has the highest precedence, and
we will use parentheses for grouping when these precedence and
associativity rules get in the way.

The language of a regular expression $\E$, $L(\E)$, is traditionally
defined through operations on
sets. Intuitively, the languages of $\Epsi$
and $a$ are singleton sets with the corresponding symbols, the
language of $\Con{e_1}{e_2}$ is given by concatenating all strings of
$L(e_1)$ with all strings of $L(e_2)$, the language of
$\Choice{e_1}{e_2}$ is the union of the languages of $e_1$ and $e_2$,
and the language of $e_1^*$ is the Kleene closure of the language of
$e_1$, that is, $L^* = \bigcup_{i=0}^\infty L^i$ where $L^0 = \{ \Epsi
\}$ and $L^i = LL^{i-1}$ for $i > 0$~\cite[p. 28]{hopcroft:automata}. 

We are interested in a semantics for regexes, the kind of regular
expressions used for pattern matching and searching, so we will 
define a {\em matching} relation for regular expressions,
$\Lr$. Informally, we will have $\Reg{\E}{xy} \Lr y$ if and only if
the expression $\E$ matches the prefix $x$ of input string $xy$. 

\begin{figure}[t]
{\small
\begin{align*}
& \textbf{Empty String} \fivespaces
{\frac{}{\Reg{\Epsi\,}{x} \Lr \Just{x}}} \mylabel{empty.1}
\tenspaces
\textbf{Character} \fivespaces
{\frac{}{\Reg{a}{ax} \Lr \Just{x}}} \mylabel{char.1}  \\ \\
& \textbf{Concatenation} \fivespaces
{\frac{\Reg{e_1}{xyz} \Lr \Just{yz} \interf \Reg{e_2}{yz} \Lr \Just{z}}
{\Reg{\Con{e_1}{e_2}}{xyz} \Lr \Just{z}}} \mylabel{con.1} \\ \\
& \textbf{Choice} \tenspaces
{\frac{\Reg{e_1}{xy} \Lr \Just{y}}
	{\Reg{\Choice{e_1}{e_2}}{xy} \Lr \Just{y}}} \mylabel{choice.1} \tenspaces
{\frac{\Reg{e_2}{xy} \Lr \Just{y}}
	{\Reg{\Choice{e_1}{e_2}}{xy} \Lr \Just{y}}} \mylabel{choice.2} \\ \\
%Repetition
& \textbf{Repetition} \fivespaces
{\frac{}
	{\Reg{\E^*}{x} \Lr \Just{x}}} \mylabel{rep.1}  \fivespaces
{\frac{\Reg{\E}{xyz} \Lr \Just{yz} \interf \Reg{\E^*}{yz} \Lr \Just{z}}
	{\Reg{\E^*}{xyz} \Lr \Just{z}}} \mbox{ , } x \neq \Epsi \mylabel{rep.2} 
\end{align*}
}
\caption{Natural semantics of relation \,$\Lr$}
\label{fig:re}
\end{figure}

Formally, we define $\Lr$ via natural semantics, using the set of inference rules in Figure~\ref{fig:re}. We have $\Reg{e}{xy} \Lr y$ if and only if we can build a proof tree for this statement using the inference rules. The rules follow naturally from the expected behavior of each expression: rule {\bf empty.1} says that $\Epsi$ matches itself and does not consume the input; rule {\bf char.1} says that a symbol matches and consumes itself if it is the beginning of the input; rule {\bf con.1} says that a concatenation uses the suffix of the first match as the input for the next; rules {\bf choice.1} and {\bf choice.2} say that a choice can match the input using either option; finally, rules {\bf rep.1} and {\bf rep.2} say that a repetition can either match $\Epsi$ and not consume the input or match its subexpression and match the repetition again on the suffix that the subexpression left.

The following lemma proves that the set of strings that expression
$\E$ matches is the language of $\E$, that is, $L(\E) = \{ x \in T^* \
| \ \exists y \ \Reg{\E}{xy} \Lr y , y \in T^*\}$.

\begin{lemma}
\label{prop:equivre}
Given a regular expression $\E$
\,and a string $x$,
for any string $y$ we have that\,
$x \in L(\E)$ \,if and only if\,
$\Reg{\E}{xy} \Lr y$.
\end{lemma}

\begin{proof}
($\Rightarrow$): By induction on the complexity of the pair
$\Tup{\E}{x}$. Given the pairs $\Tup{e_1}{x_1}$ \,and\,
$\Tup{e_2}{x_2}$, the first pair is more complex than the second one
if and only if either $e_2$ is a proper subexpression of $e_1$ or $e_1
= e_2$ and $|x_1| > |x_2|$. The base cases are $\Tup{\Epsi}{\Epsi}$ and
$\Tup{a}{a}$, and their proofs follow by application of rules {\bf
  empty.1} and {\bf char.1}, respectively. Cases
$\Tup{\Con{e_1}{e_2}}{x}$ and $\Tup{\Choice{e_1}{e_2}}{x}$ use a
straightforward application of the induction hypothesis on the
subexpressions, followed by application of rule {\bf con.1} or one of
the {\bf choice} rules. Case $\Tup{e^*}{\Epsi}$ follows directly from
rule {\bf rep.1}, while for case $\Tup{e^*}{x}$, where $x \neq \Epsi$,
we know by the definition of the Kleene closure that $x \in L^i(e_1)$
with $i > 0$, where $L^i(e_1)$ is $L(e_1)$ concatenated with itself
$i$ times. This means that we can decompose $x$ into $x_1x_2$, with a
non-empty $x_1$, where
$x_1 \in L(e_1)$ and $x_2 \in L^{i-1}(e_1)$. Again by the definition
of the Kleene closure this means that $x_2 \in L(e_1^*)$. The proof
now follows by the induction hypothesis on $\Tup{e_1}{x_1}$ and
$\Tup{e_1^*}{x_2}$ and an application of rule {\bf rep.2}.

($\Leftarrow$): By induction on the height of the proof tree for $\Reg{e}{xy} \Lr y$. Most cases are straightforward; the interesting case is when the proof tree concludes with rule {\bf rep.2}. By the induction hypothesis we have that $x \in L(e_1)$ and $y \in L(e_1^*)$. By the definition of the Kleene closure we have that $y \in L^i(e_1)$, so $xy \in L^{i+1}(e_1)$ and, again by the Kleene closure, $xy \in L(e_1^*)$, which concludes the proof.
\end{proof}

Salomaa~\cite{salomaa} developed a complete axiom system for regular
expressions, where any valid equation involving regular expressions
can be derived from the axioms. The axioms of system $F_1$ are:

\begin{align}
\Chre{e_1}{(\Chre{e_2}{e_3})} & \Myeq \Chre{(\Chre{e_1}{e_2})}{e_3} \\
e_1(e_2e_3) & \Myeq (e_1e_2)e_3 \\
\Chre{e_1}{e_2} & \Myeq \Chre{e_2}{e_1} \\
e_1(\Chre{e_2}{e_3}) & \Myeq \Chre{e_1e_2}{e_1e_3} \\
(\Chre{e_1}{e_2})e_3 & \Myeq \Chre{e_1e_3}{e_2e_3} \\
\Chre{e}{e} & \Myeq e \\
\Epsi e & \Myeq e \\
\emptyset e & \Myeq \emptyset \\
\Chre{e}{\emptyset} & \Myeq e \\
e^* & \Myeq \Chre{\Epsi}{e^* e} \\
e^* & \Myeq (\Chre{\Epsi}{e})^*
\end{align}

Salomaa's regular expressions do not have the $\Epsi$ case; the
original axioms use $\emptyset^*$, which has the same meaning, as the
only possible proof trees for $\emptyset^*$ use {\bf rep.1}.
The following lemma shows that these axioms are valid under our
semantics for regular expressions, if we take $e_1 = e_2$ to mean that
$e_1$ and $e_2$ match the same sets of strings. 

\begin{lemma}
\label{lemma:salomaa}
For each of the axioms of $F_1$, if $l$ is the expression on the left
side and $r$ is the expression on the right side, we have that
$\Reg{l}{xy} \Lr y$ if and only if $\Reg{r}{xy} \Lr y$.
\end{lemma}

\begin{proof}
Trivially true for axiom 8, as there are no proof trees for either the
left or right sides of this axiom. For axioms 1 to 7 and for axiom 9
it is straightforward to use the subtrees of the proof tree of one
side to build a proof tree for the other side. We can prove the validity of axiom
11 by an straightforward induction
on the height of the proof trees for each side.

For axiom 10, we need to prove the identity $a^*a = a a^*$, by
induction on the heights of the proof trees. From this identity the
left side of axiom 10 follows by taking the subtrees of {\bf rep.2}
and combining them with {\bf con.1} into a tree for $a a^*$, which
means we have a tree for $a^* a$ that we can use to build a tree for
the right side using {\bf choice.2}. The right side follows from
getting a tree for $a a^*$ from a tree for $a^* a$ using the identity,
then taking its subtrees and using {\bf rep.2} to get a tree for $a^*$. 
\end{proof}

Parsing expression grammars borrow the syntax of regular
expressions. A {\em parsing expression} is also defined inductively,
extending the inductive definition of regular expressions with two new cases,
$A$ for a non-terminal, and $!e$ for a {\em not-predicate} of
expression $e$. A PEG $G$ is a tuple $(V, T, P, p_S)$ where $V$ is the set of
non-terminals, $T$ is the alphabet (set of terminals), $P$ is a
function from $V$ to parsing expressions, and $p_S$ is the parsing
expression that the PEG matches (its starting parsing expression). We will use the notation $G[p]$ for a
grammar derived from $G$ where $p_S$ is replaced by $p$ while keeping
$V$, $T$, and $P$ the same. We will refer to both regular expressions
and parsing expressions as just expressions, letting context
disambiguate between both kinds.

While the syntax of parsing expressions is similar to the syntax of regular
expressions, the behavior of the choice and repetition operators is
very different. Choice in PEGs is {\em ordered}; a PEG will only try
to match the right side of a choice if the left side cannot match any
prefix of the input. Repetition in PEGs is {\em possessive}; a repetition
will always consume as much of the input as it can
match, regardless of whether this leads to failure or a shorter match
for the whole pattern\footnote{Possessive repetition is a consequence of ordered choice,
  as $e^*$ is the same as expression $A$ where $A$ is a fresh
  non-terminal and $P(A) = \Choice{eA}{\Epsi}$.}. To formally define
ordered choice and possessive repetition we also need a way to express
that an expression does not match a prefix of the input, so we need to
introduce $\Nothing$ as a possible outcome of a match.

%%% Operational Semantics of PEGs %%%
\begin{figure}[t]
{\scriptsize
\begin{align*}
%Matching empty string
& \textbf{Empty String} \fivespaces
{\frac{}{\Matg{\Epsi}{x} \Lp \Just{x}}} \mylabel{empty.1} 
\tenspaces
%Variable
\textbf{Non-terminal} \fivespaces
{\frac{\Matg{P(A)}{x} \Lp \Just{\Anyx}}
	{\Matg{A}{x} \Lp \Just{\Anyx}}}    \mylabel{var.1}^+  \\ \\ 
%Matching a given terminal
& \textbf{Terminal} \;\;
{\frac{}{\Matg{a}{ax} \Lp \Just{x}}} \mylabel{char.1} \;\;
{\frac{}{\Matg{b}{ax} \Lp \Nothing}} \mbox{ , } b \neq a \mylabel{char.2}^+ \;\;
{\frac{}{\Matg{a}{\Epsi} \Lp \Nothing}} \mylabel{char.3}^+  \\ \\
%Concatenation
& \textbf{Concatenation}
\fivespaces
{\frac{\Matg{p_1}{xy} \Lp \Just{y} \interf \Matg{p_2}{y} \Lp \Just{\Anyx}}
	{\Matg{\Con{p_1}{p_2}}{xy} \Lp \Just{\Anyx}}} \mylabel{con.1} \fivespaces
{\frac{\Matg{p_1}{x} \Lp \Nothing}
	{\Matg{\Con{p_1}{p_2}}{x} \Lp \Nothing}} \mylabel{con.2}^+ \\ \\
%Ordered Choice
& \textbf{Ordered Choice} 
\;\;\;
{\frac{\Matg{p_1}{xy} \Lp \Just{y}}
	{\Matg{\Choice{p_1}{p_2}}{xy} \Lp \Just{y}}} \mylabel{choice.1} \;\;
{\frac{\Matg{p_1}{x} \Lp \Nothing \interf \Matg{p_2}{x} \Lp \Just{\Anyx}}
	{\Matg{\Choice{p_1}{p_2}}{x} \Lp \Just{\Anyx}}} \mylabel{choice.2}^* \\ \\ 
%Repetition
& \textbf{Repetition} \fivespaces
{\frac{\Matg{p}{x} \Lp \Nothing}
	{\Matg{p^*}{x} \Lp \Just{x}}} \mylabel{rep.1}^* \fivespaces
{\frac{\Matg{p}{xyz} \Lp \Just{yz} \interf \Matg{p^*}{yz} \Lp \Just{z}}
	{\Matg{p^*}{xyz} \Lp \Just{z}}} \mylabel{rep.2} \\ \\
%Not
& \textbf{Not Predicate} \tenspaces \fivespaces
{\frac{\Matg{p}{x} \Lp \Nothing} 
	{\Matg{!p}{x} \Lp \Just{x}}} \mylabel{not.1}^+ \tenspaces
{\frac{\Matg{p}{xy} \Lp \Just{y}}
	{\Matg{!p}{xy} \Lp \Nothing}} \mylabel{not.2}^+
\end{align*}
\caption{Definition of Relation $\Lp$ through Natural Semantics}
\label{fig:matchpeg}
}
\end{figure}

Figure~\ref{fig:matchpeg} gives the definition of $\Lp$, the matching
relation for PEGs. As with regular expressions, we say that
$\Mat{G}{xy} \Lp y$ to express that the grammar $G$ matches the prefix
$x$ of input string $xy$, and the set of strings that a PEG
matches is its language: $L(G) = \{ x \in T \ | \ \exists y \, \Mat{G}{xy} \Lp y,\ y
\in T^* \}$. 

We mark with a $*$ the rules that have been changed from Figure~\ref{fig:re}, and
mark with a $+$ the rules that were added. Unmarked rules are
unchanged from Figure~\ref{fig:re}, except for the trivial change of
adding the parameter $G$ to the relation. We have six new rules, and two
changed rules. New rules {\bf char.2} and {\bf char.3} generate
$\Nothing$ in the case that the expression cannot match the symbol in
the beginning of the input. New rule {\bf con.2} says that a
concatenation fails if its left side fails. New rule {\bf var.1} says
that to match a non-terminal we have to match the parsing expression
associated with it in this grammar's function from non-terminals to
parsing expressions (the non-terminal's ``right-hand side'' in the grammar). New rules
{\bf not.1} and {\bf not.2} say that a not predicate never consumes
input, but fails if its subexpression matches a prefix of the input. 

The change in rule {\bf con.1} is trivial and only serves to propagate
$\Nothing$, so we do not consider it an actual change. The changes to
rules {\bf choice.2} and {\bf rep.1} are what actually implements
ordered choice and possessive repetition, respectively. Rule {\bf
  choice.2} says that we can only match the right side of the choice
if the left side fails, while rule {\bf rep.1} says that a repetition
only stops if we try to match its subexpression and fail. 

It is easy to see that PEGs are deterministic; that is, a given PEG G
can only have a single result (either $\Nothing$ or a suffix of $x$)
for some input $x$, and only a single proof tree for this
result. If the PEG $G$ always yields a result for any input in $T^*$
then we say that $G$ is {\em complete}~\cite{ford:peg}. PEGs that are
not complete include any PEG that has left recursion and PEGs with
repetitions $e^*$ where $e$ matches the empty string. From now on we
will assume that any PEG we consider is complete unless stated
otherwise. The completeness of a PEG can be proved syntactically~\cite{ford:peg}.

The syntax of the expressions that form a PEG are a superset of the
syntax of regular expressions, so syntactically any regular expression
$e$ has a corresponding PEG $G_e = (V, T, P, e)$, where $V$ and $P$
can be anything. We can prove that $L(G_e) \subseteq L(\E)$ by a
simple induction on the height of proof trees for $\Mat{G_e}{xy} \Lp
y$, but it is easy to show examples where $L(G_e)$ is a proper subset
of $L(\E)$, so the regular expression and its corresponding PEG have different languages.

For example, expression $\Chre{a}{ab}$ has the language $\{ \Ch{a},\,\Ch{ab}\}$
as a regular expression but $\{ \Ch{a} \}$ as a PEG, because on an
input with prefix $\Ch{ab}$ the left side of the choice always matches
and keeps the right side from being tried. The same happens with
expression $\Con{a}{(\Chre{b}{bb})}$, which has language $\{
\Ch{ab},\, \Ch{abb} \}$ as a regular expression and $\{ \Ch{ab} \}$ as
a PEG, and on inputs with prefix $\Ch{abb}$ the left side of the
choice keeps the right side from matching.

The behavior of the PEGs actually match the behavior of the {\em
  regexes} $\Chre{a}{ab}$ and $\Con{a}{(\Chre{b}{bb})}$ on
Perl-compatible regex engines. These engines will always match
$\Chre{a}{ab}$ with just the first {\tt a} on subjects starting with
{\tt ab}, and always match $\Con{a}{(\Chre{b}{bb})}$ with {\tt ab} on
subjects starting with {\tt abb}, unless the order
of the alternatives is reversed. 

A different situation happens with expression
$\Con{(\Chre{a}{aa})}{b}$. As a regular expression, its language is
$\{ \Ch{ab},\, \Ch{aab} \}$ while as a PEG it is $\{ \Ch{ab} \}$, but
the PEG fails on an input with prefix $\Ch{aab}$. Regex engines will
backtrack and try the second alternative when $b$ fails to match the
second {\tt a}, and will also match {\tt aab}, highlighting the
difference between the unrestricted backtracking of common regex
implementations and PEGs' restricted backtracking.

If we take the previous regular expression, $\Con{(\Chre{a}{aa})}{b}$,
and distribute $b$ over the two alternatives, we have
$\Chre{ab}{aab}$. This expression now has the same language when we
interpret it either as a regular expression, as a PEG, or as a regex.

If we change $\Chre{a}{ab}$ and $\Con{a}{(\Chre{b}{bb})}$ so they have
the {\em prefix property}\footnote{There are no distinct strings $x$
  and $y$ in the language such that $x$ is a prefix of $y$.}, by
adding an end-of-input marker \$, we have the expressions
$\Con{(\Chre{a}{ab})}{\$}$ and
$\Con{a}{\Con{(\Chre{b}{bb})}{\$}}$. Now their languages as regular
expressions are $\{ \Ch{a\$},\, \Ch{ab\$} \}$ and $\{ \Ch{ab\$},\,
\Ch{abb\$} \}$, respectively, but the first expression fails as a PEG
on an input with prefix $\Ch{ab}$ and the second expression fails as a
PEG on an input with prefix $\Ch{abb}$.

Both $\Con{(\Chre{a}{ab})}{\$}$  and
$\Con{a}{\Con{(\Chre{b}{bb})}{\$}}$ match, as regexes, the same set of
strings that form their languages as regular expressions. They are in
the form $\Con{(\Chre{e_1}{e_2})}{e_3}$, like
$\Con{(\Chre{a}{aa})}{b}$, so we can distribute $e_3$ over the choice
to obtain $\Chre{\Con{e_1}{e_3}}{\Con{e_2}{e_3}}$. If we do that the
two expressions become $\Chre{a\$}{ab\$}$ and
$\Con{a}{(\Chre{b\$}{bb\$})}$, respectively, and they now have the
same language either as a regular expression, as a PEG, or as a regex.

We will say that a PEG $G$ and a regular expression $\E$ over the same
alphabet $T$ are {\em equivalent} if the following conditions hold for every input string $xy$:
\begin{eqnarray}
\Mat{G}{xy} \Lp y  & \Rightarrow & \Reg{\E}{xy} \Lr y \\
\Reg{\E}{xy} \Lr y & \Rightarrow & \Mat{G}{xy} \Nlp \Nothing
\end{eqnarray}

That is, a PEG $G$ and a regular expression $\E$ are equivalent if
$L(G) \subseteq L(\E)$ and $G$ does not fail for any string with a
prefix in $L(\E)$. In the examples above, regular expressions
$\Chre{a}{ab}$, $\Con{a}{(\Chre{b}{bb})}$, $\Chre{a\$}{ab\$}$,
$\Con{a}{(\Chre{b\$}{bb\$})}$, and $\Chre{ab}{aab}$ are all equivalent
with their corresponding PEGs, while $\Con{(\Chre{a}{ab})}{\$}$,
$\Con{a}{\Con{(\Chre{b}{bb})}{\$}}$, and $\Con{(\Chre{a}{aa})}{b}$ are
not.

Informally, a PEG and a regular expression will be equivalent if the
PEG matches the same strings as the regular expression when the
regular expression is viewed as a regex under the common
``leftmost-first'' semantics of Perl-compatible regex
implementations. If a regular expression can match different prefixes
of the same subject, due to the non-determinism of the choice and
repetition operations, the two conditions of equivalence guarantee
that an equivalent PEG will match one of those prefixes.

Regexes are deterministic, and will also match just a single prefix of
the possible prefixes a regular expression can match. Our
transformation will preserve the ordering among choices, so the prefix an
equivalent PEG obtained with our transformation matches will be the
same prefix a regex matches.

While equivalence is enough to guarantee that a PEG will give the same
results as a regex, equivalence together with the prefix property
yields the following lemma:

\begin{lemma}
\label{lemma:repeglang}
If a regular expression $\E$ with the prefix property and a PEG $G$ are equivalent then $L(G) = L(\E)$.
\end{lemma}

\begin{proof}
As our first condition of equivalence says that $L(G) \subseteq L(\E)$, we just need to prove that $L(\E) \subseteq L(G)$. Suppose there is a string $x \in L(\E)$; this means that $\Reg{\E}{xy} \Lr y$ for any $y$. But from equivalence this means that $\Mat{G}{xy} \not\Lp \Nothing$. As $G$ is complete, we have $\Mat{G}{xy} \Lp y^\prime$. By equivalence, the prefix of $xy$ that $G$ matches is in $L(\E)$. Now $y$ cannot be a proper suffix of $y^\prime$ nor $y^\prime$ a proper suffix of $y$, or the prefix property would be violated. This means that $y^\prime = y$, and $x \in L(G)$, completing the proof.
\end{proof}

We can now present an overview on how we will transform a regular
expression $\E$ into an equivalent PEG. We first need to transform
subexpressions of the form $\Con{(\Chre{e_1}{e_2})}{e_3}$ to
$\Chre{e_1e_3}{e_2e_3}$. We do not need to actually duplicate $e_3$ in
both sides of the choice, potentially causing an explosion in the size of the resulting
expression, but can introduce a fresh non-terminal $X$ with $P(X) =
e_3$, and distribute $X$ to transform $\Con{(\Chre{e_1}{e_2})}{e_3}$
into $\Chre{e_1X}{e_2X}$.

Transforming repetition is trickier, but we just have to remember that
$e_1^* e_2 \equiv \Con{(\Choice{\Con{e_1}{e_1^*}}{\Epsi})}{e_2} \equiv
\Choice{(\Con{e_1}\Con{e_1^*}{e_2})}{e_2}$. Naively transforming the
first expression to the third does not work, as we end up with
$e_1^*e_2$ in the expression again, but we can add a fresh
non-terminal $A$ to the PEG with $P(A) = \Choice{e_1 A}{e_2}$ and then
replace $e_1^* e_2$ with $A$ in the original expression. The end
result of repeatedly applying these two transformation steps until we
reach a fixed point will be a parsing
expression that is equivalent to the original regular expression. 

As an example, let us consider the regular expression $b^*b\$$. Its
language is $\{ \Ch{b\$},\, \Ch{bb\$},\, \ldots \}$, but when
interpreted as a PEG the language is $\emptyset$, due to possessive
repetition. If we transform the original expression into a PEG with
starting parsing expression $A$ and $P(A) = \Choice{bA}{b\$}$, it will
have the same language as the original regular expression;
 for example, given the input \Ch{bb\$}, this PEG matches
the first \Ch{b} through subexpression $b$ of $\Con{b}{A}$,
and then $A$ tries to match the rest of the input, \Ch{b\$}. So, once more
subexpression $b$ of $\Con{b}{A}$ matches \Ch{b} and then $A$ tries to
match the rest of the input, $\Ch{\$}$. Since both $\Con{b}{A}$ and
$b\$$ fail to match $\Ch{\$}$, $A$ fails, and thus $\Con{b}{A}$ fails
for input $\Ch{b\$}$. Now we try $b\$$, which successfully matches $\Ch{b\$}$, and the
complete match succeeds.

If we now consider the regular expression $b^*b$, which has $\{
\Ch{b},\, \Ch{bb},\, \ldots \}$ as its language, we have that it also
has the empty set as its language if we interpret it as a PEG. A PEG
with starting parsing expression $A$ and $P(A) = \Choice{bA}{b}$ also
has $\{ \Ch{b},\, \Ch{bb},\, \ldots \}$ as its language, but with an
important difference: the regular expression can match just the first
{\tt b} of a subject starting with {\tt bb}, but the PEG will match
both, and any other ones that follow, so we do not have $\Reg{\E}{xy}
\Lr y$ implying $\Mat{G}{xy} \Lp y$. But the behavior of the PEG
corresponds the behavior of regex engines, which use {\em greedy}
repetition, where a repetition will always match as much as it can
while not making the rest of the expression fail.

The next section formalizes our transformation, and proves that for
any regular expression $e$ it will give a PEG that is equivalent to
$e$, that is, it will yield a PEG that recognizes the same language as $e$ if it has the prefix property.

\section{Transforming Regular Expressions to PEGs}
\label{sec:equiv}

This section presents function $\Pi$, a formalization of the transformation we outlined in the previous section. The function $\Pi$ transforms a regular expression $\E$ using a PEG $\Gk$ that is equivalent to a regular expression $\Ek$ to yield a PEG that is equivalent to the regular expression $\Con{\E}{\Ek}$.

The intuition behind $\Pi$ is that $\Gk$ is a {\em continuation} for the regular expression $\E$, being what should be matched after matching $\E$. We use this continuation when transforming choices and repetitions to do the transformations of the previous section; for a choice, the continuation is distributed to both sides of the choice. For a repetition, it is used as the right side for the new non-terminal, and the left side of this non-terminal is the transformation of the repetition's subexpression with the non-terminal as continuation.

For a concatenation, the transformation is the result of transforming the right side with $\Gk$ as continuation, then using this as continuation for transforming the left side. This lets the transformation of $\Con{(\Chre{e_1}{e_2})}{e_3}$ work as expected: we transform $e_3$ and then use the PEG as the continuation that we distribute over the choice. 

We can transform a standalone regular expression $\E$ by passing a PEG with $\Epsi$ as starting expression as the continuation; this gives us a PEG that is equivalent to the regular expression $\Con{\E}{\Epsi}$, or $\E$.
  
\begin{figure}[t]
{\small
\begin{align*}
& \Rp{\Epsi}{\Gk}  \; = \;  \Gk  \fivespaces
\Rp{a}{\Gk}  \; = \; \Peg{G_k}{\Con{a}{\Pk}}  \fivespaces
\Rp{\Con{e_1}{e_2}}{\Gk}  \; = \; \Rp{e_1}{\Rp{e_2}{\Gk}}  \\
& \Rp{\Choice{e_1}{e_2}}{\Gk}  \; = \;
     \Peg{G_2}{\Choice{p_1}{p_2}}  \mbox{, where} \,\,
      G_2 = \Grm{V_2}{T}{P_2}{p_2} = \Rp{e_2}{\Grm{V_1}{T}{P_1}{p_k}} \\
&  \tenspaces \tenspaces \tenspaces \tenspaces \tenspaces \fivespaces  \; \mbox{and} \;
 	 \;\;\, \Grm{V_1}{T}{P_1}{p_1} = \Rp{e_1}{\Gk} \\
& \Rp{e_1^*}{\Gk}  \; = \; G \;
		 \mbox{, where} \;\,
G \;=\; \Grm{V_1}{T}{P_1 \cup \{A \arrow \Chre{p_1}{\Pk}\}}{A} \mbox{ with }A \notin V_k \mbox{ and }\\
&		 \tenspaces \tenspaces \tenspaces \tenspaces  \!  
     \Grm{V_1}{T}{P_1}{p_1} \;=\; 
                            \Rp{e_1}{\Grm{V_k \cup \{A\}}{T}{P_k}{A}}
\end{align*}
\caption{Definition of Function $\Pi$, where $\Gk = \Grm{\Vk}{T}{P_k}{\Pk}$}
\label{fig:defpi}
}
\end{figure}

Figure~\ref{fig:defpi} has the definition of function $\Pi$. Notice
how repetition introduces a new non-terminal, and the transformation
of choice has to take this into account by using the set of
non-terminals and the productions of the result of transforming one
side to transform the other side, so there will be no overlap. Also
notice how we transform a repetition by transforming its body using
the repetition itself as a continuation (through the introduction of a
fresh non-terminal), then building a choice between the transformation
and the body and the continuation of the repetition. The
transformation process is bottom-up and right-to-left.

We will show the subtler points of transformation $\Pi$ with some
examples. In the following
discussion, we use the alphabet $T = \{\Ch{a},\,\Ch{b},\,\Ch{c}\}$,
and the continuation grammar $\Gk =
\Grm{\emptyset}{T}{\emptyset}{\Epsi}$ that is equivalent to the
regular expression $\Epsi$. In our first example, we use the
regular expression $\Con{(\Chre{a}{\Chre{b}{c}})^*}
    {\Con{a\,}
    {(\Chre{a}{\Chre{b}{c}})^*}}$, which matches an input that has at least
one \Ch{a}.

We first transform the second repetition by evaluating $\Rp{(\Chre{a}{\Chre{b}{c}})^*}{\Gk}$; we first transform $\Chre{a}{\Chre{b}{c}}$ with a new non-terminal $A$ as continuation, yielding the PEG $\Chre{aA}{\Chre{bA}{cA}}$, then combine it with $\Epsi$ to yield the PEG $A$ where $A$ has the production below:
\[
A \arrow\; \Choexe{\Con{a}{A}}
                {\Choexe{\Con{b}{A}}
                        {\Choexe{\Con{c}{A}}
                        {\Epsi}}}
\]

Next is the concatenation with $a$, yielding the PEG $aA$. We then use this PEG as continuation for transforming the first repetition. This transformation uses a new non-terminal $B$ as a continuation for transforming $\Chre{a}{\Chre{b}{c}}$, yielding $\Chre{aB}{\Chre{bB}{cB}}$, then combines it with $aA$ to yield the PEG $B$ with the productions below:

\[
B \arrow\; \Choexe{\Con{a}{B}}
                {\Choexe{\Con{b}{B}}
                        {\Choexe{\Con{c}{B}}
                        {\Con{a}{A}}}}
\tenspaces
A \arrow\; \Choexe{\Con{a}{A}}
                {\Choexe{\Con{b}{A}}
                        {\Choexe{\Con{c}{A}}
                        {\Epsi}}}
\]

When the original regular expression matches a given input, we do not
know how many \Ch{a}'s the first repetition matches, because the
semantics of regular expressions is non-deterministic. Regex
implementations commonly resolve ambiguities in repetitions
 by the longest match rule, where the first repetition will match all
 but the last \Ch{a} of the input. PEGs are deterministic by
 construction, and the PEG generated by $\Pi$ obeys the longest match
 rule. The alternative $\Con{a}{A}$ of non-terminal $B$ will only be
 tried if all the alternatives fail, which happens in the end of the
 input. The PEG then backtracks until the last \Ch{a} is found, where
 it matches the last \Ch{a} and proceeds with non-terminal $A$.

The regular expression $\Con{(\Chre{b}{c})^*}{\Con{a\,}
    {(\Chre{a}{\Chre{b}{c}})^*}}$ defines the same language as the
  regular expression of the first example, but without the
  ambiguity. Now $\Pi$ with continuation $\Gk$ yields the following PEG $B$:
\[
B \arrow\; \Choexe{\Con{b}{B}}
                {\Choexe{\Con{c}{B}}
                {\Con{a}{A}}}
\tenspaces
A \arrow\; \Choexe{\Con{a}{A}}
                {\Choexe{\Con{b}{A}}
                        {\Choexe{\Con{c}{A}}
                        {\Epsi}}}
\]

Although the productions of this PEG and the previous one match the
same strings, the second PEG is more efficient, as it will not have to
reach the end of the input and then backtrack until finding the last
\Ch{a}. This is an example on how we can use our semantics and the
transformation $\Pi$ to reason about the behavior of a regex. The
relative efficiency of the two PEGs is an artifact of the semantics,
while the relative efficiency of the two regexes depends on how a
particular engine is implemented. In a backtracking implementation it
will depend on what ad-hoc optimizations the implementation makes, in
an automata implementation they both will have the same relative
efficiency, at the expense of the implementation lacking the
expressive power of some regex extensions.

The expressions in the two previous examples are {\em well-formed}. A
regular expression $\E$ is well-formed if it does not have a
subexpression $e_i^*$ where $\Epsi \in L(e_i)$. If $\E$ is a
well-formed regular expression and $\Gk$ is a complete PEG then
$\Rp{\E}{\Gk}$ is also complete. In Section~\ref{sec:left} we will
show how to mechanically obtain a well-formed regular expression that
recognizes the same language as a non-well-formed regular expression
while preserving its overall structure.

We will now prove that our transformation $\Pi$ is correct, that is, if $\E$ is a well-formed regular expression and $\Gk$ is a PEG equivalent to a regular expression $\Ek$ then $\Rp{\E}{\Gk}$ is equivalent to $\Con{\E}{\Ek}$. The proofs use a small technical lemma: each production of PEG $\Gk$ is also in PEG $\Rp{\E}{\Gk}$, for any regular expression $\E$. This lemma is straightforward to prove by structural induction on $\E$.

We will prove each property necessary for equivalence separately; equivalence will then be a direct corollary of those two proofs. To prove the first property we need an auxiliary lemma that states that the continuation grammar is indeed a continuation, that is if the PEG $\Rp{\E}{\Gk}$ matches a prefix $x$ of a given input $xy$ then we can split $x$ into $v$ and $w$ with $x = vw$ and $\Gk$ matching $w$.

\begin{lemma}
\label{lemma:matchcon}
Given a well-formed regular expression $\E$, a PEG $\Gk$, and an input string $xy$, if $\Mat{\Rp{\E}{\Gk}}{xy} \Lp y$ then there is a suffix $w$ of $x$ such that \, $\Mat{\Gk}{wy} \Lp y$.
\end{lemma}

\begin{proof}
By induction on the complexity of the pair $\Tup{\E}{xy}$. The interesting case is $e^*$. In this case $\Rp{e^*}{\Gk}$ gives us a grammar $G = \Grm{V_1}{T}{P}{A}$, \,where\, $A \arrow \Choice{p_1}{\Pk}$.
By {\bf var.1} we know that $\Mat{\Peg{G}{\Choice{p_1}{\Pk}}}{xy} \Lp y$. There are now two subcases to consider, {\bf choice.1} and {\bf choice.2}.

For subcase {\bf choice.2}, we have $\Matg{\Pk}{xy} \Lp y$. But then we have that $\Matk{\Pk}{xy} \Lp y$ because any non-terminal that $\Pk$ uses to match $xy$ is in both $G$ and $\Gk$ and has the same production in both. The string $xy$ is a suffix of itself, and $p_k$ is the starting expression of $\Gk$, closing this part of the proof.

For subcase {\bf choice.1} we have $\Mat{\Rp{e}{\Rp{e^*}{\Gk}}}{xy} \Lp y$,
and by the induction hypothesis $\Mat{\Rp{e^*}{\Gk}}{wy} \Lp y$.
We can now use the induction hypothesis again, on the length of the input, as $w$ must be a proper suffix of $x$. We conclude that $\Mat{\Gk}{\Wp y} \Lp y$ for a suffix $w^\prime$ of $w$, and so a suffix of $x$, ending the proof.
\end{proof}

The following lemma proves that if the first property of equivalence holds between a regular expression $\Ek$ and a PEG $\Gk$ then it will hold for $\Con{\E}{\Ek}$ and $\Rp{\E}{\Gk}$ given a regular expression $\E$.

\begin{lemma}
\label{lemma:reas}
Given two well-formed regular expressions $\E$ and $\Ek$ and a PEG $\Gk$, where $\Mat{\Gk}{wy} \Lp y \Rightarrow \Reg{\Ek}{wy} \Lr y$, if $\Mat{\Rp{\E}{\Gk}}{vwy} \Lp y$ then $\Reg{\Con{\E}{\Ek}}{vwy} \Lr y$.
\end{lemma}

\begin{proof}
By induction on the complexity of the pair $\Tup{\E}{vwy}$. The interesting case is $e^*$. In this case, $\Rp{e^*}{\Gk}$ gives us a
PEG $G = \Grm{V_1}{T}{P}{A}$, where\, $A \arrow \Choice{p_1}{\Pk}$.
By {\bf var.1} we know that $\Matg{\Choice{p_1}{\Pk}}{vwy} \Lp y$. There are now two subcases, {\bf choice.1} and {\bf choice.2} of $\Lp$.

For subcase {\bf choice.2}, we can conclude that $\Gk \, vwy \Lp y$ because $p_k$ is the starting expression of $\Gk$ and any non-terminals it uses have the same production both in $G$ and $\Gk$. We now have $\Reg{\Ek}{vwy} \Lr y$. By {\bf choice.2} of $\Lr$ we have $\Reg{\Choice{\Con{e}{\Con{e^*}{\Ek}}}{\Ek}}{vwy} \Lr y$, but $\Choice{\Con{e}{\Con{e^*}{\Ek}}}{\Ek}
\equiv \Con{e^*}{\Ek}$, so $\Reg{\Con{e^*}{\Ek}}{vwy} \Lr y$, ending this part of the proof.

For subcase {\bf choice.1}, we have $\Mat{\Rp{e}{\Rp{e^*}{\Gk}}}{vwy} \Lp y$,
and by Lemma~\ref{lemma:matchcon} we have $\Mat{\Rp{e^*}{\Gk}}{wy} \Lp y$. The string $v$ is not empty, so we can use the induction hypothesis and Lemma~\ref{lemma:matchcon} again to conclude $\Reg{\Con{e^*}{\Ek}}{wy} \Lr y$. Then we use the induction hypothesis on $\Mat{\Rp{e}{\Rp{e^*}{\Gk}}}{vwy} \Lp y$ to conclude $\Reg{\Con{\E}{\Con{\E^*}{\Ek}}}{vwy} \Lr y$. We can now use rule {\bf choice.1} of $\Lr$ to get $\Reg{\Choice{\Con{e}{\Con{e^*}{\Ek}}}{\Ek}}{vwy} \Lr y$, but $\Choice{\Con{e}{\Con{e^*}{\Ek}}}{\Ek}
\equiv \Con{e^*}{\Ek}$, so $\Reg{\Con{e^*}{\Ek}}{vwy} \Lr y$, ending the proof.
\end{proof}

The following lemma proves that if the second property of equivalence holds between a regular expression $\Ek$ and a PEG $\Gk$ then it will hold for $\Con{\E}{\Ek}$ and $\Rp{\E}{\Gk}$ given a regular expression $\E$.

\begin{lemma}
\label{lemma:prop2}
Given well-formed regular expressions $\E$ and $\Ek$ and a PEG $\Gk$, where Lemma~\ref{lemma:reas} holds and we have $\Reg{\Ek}{wy} \Lr y \Rightarrow \Mat{\Gk}{wy} \Nlp \Nothing$, if $\Reg{\Con{\E}{\Ek}}{vwy} \Lr y$ then $\Mat{\Rp{\E}{\Gk}}{vwy} \Nlp \Nothing$.
\end{lemma}

\begin{proof}
By induction on the complexity of the pair $\Tup{\E}{vwy}$. The interesting case is $e^*$. We will use again the equivalence
$\Con{e^*}{\Ek} \equiv \Choice{\Con{e}{\Con{e^*}{\Ek}}}{\Ek}$.
There are two subcases, {\bf choice.1} and {\bf choice.2} of $\Lr$.

For subcase {\bf choice.1}, we have that $e$ matches a prefix of $vwy$ by rule {\bf con.1}. As $e^*$ is well-formed this prefix is not empty, so $\Reg{\Con{e^*}{\Ek}}{v^\prime wy} \Lr y$ for a proper suffix $v^\prime$ of $v$. By the induction hypothesis we have $\Mat{\Rp{e^*}{\Gk}}{v^\prime wy} \Nlp \Nothing$, and by induction hypothesis again we get $\Mat{\Rp{e}{\Rp{e^*}{\Gk}}}{vwy} \Nlp \Nothing$. This PEG is complete, so we can conclude $\Mat{\Peg{\Rp{e^*}{\Gk}}{\Choice{p_1}{\Pk}}}{vwy} \Nlp \Nothing$ using rule {\bf choice.1} of $\Lp$, and then $\Mat{\Rp{e^*}{\Gk}}{vwy} \Nlp \Nothing$ by rule {\bf var.1}, ending this part of the proof.

For subcase {\bf choice.2}, we can assume that there is no proof tree for the statement $\Reg{\Con{e}{\Con{e^*}{\Ek}}}{vwy} \Lr y$, or we could reduce this subcase to the first one by using {\bf choice.1} instead of {\bf choice.2}. Because $\Rp{e}{\Rp{e^*}{\Gk}}$ is complete we can use modus tollens of Lemma~\ref{lemma:reas} to conclude that $\Mat{\Rp{e}{\Rp{e^*}{\Gk}}}{vwy} \Lp \Nothing$. We also have $\Reg{\Ek}{vwy} \Lr y$, so $\Mat{\Gk}{vwy} \Nlp \Nothing$. Now we can use rule {\bf choice.2} of $\Lp$ to conclude $\Matg{\Choice{p_1}{\Pk}}{vwy} \Nlp \Nothing$, and then $\Mat{\Rp{e^*}{\Gk}}{vwy} \Nlp \Nothing$ by rule {\bf var.1}, ending the proof.
\end{proof}

The correctness lemma for $\Pi$ is a corollary of the two previous lemmas:

\begin{lemma}
\label{lemma:equiv}
Given well-formed regular expressions $\E$ and $\Ek$ and a PEG $\Gk$, where $\Ek$ and $\Gk$ are equivalent, then $\Rp{\E}{\Gk}$ and $\Con{\E}{\Ek}$ are equivalent.
\end{lemma}

\begin{proof}
The proof that first property of equivalence holds for $\Rp{\E}{\Gk}$ and $\Con{\E}{\Ek}$ follows from the first property of equivalence for $\Ek$ and  $\Gk$ plus Lemma~\ref{lemma:reas}. The proof that the second property of equivalence holds follows from the first property of equivalence for $\Rp{\E}{\Gk}$ and $\Con{\E}{\Ek}$, the second property of equivalence for $\Ek$ and $\Gk$, plus Lemma~\ref{lemma:prop2}.
\end{proof}

A corollary of the previous lemma combined with
Lemma~\ref{lemma:repeglang} is that $L(\Con{\E}{\$}) =
L(\Rp{\E}{\$})$, proving that our transformation can yield a PEG that
recognizes the same language as any well-formed regular expression
$\E$ just by using an end-of-input marker, even if the language of
$\E$ does not have the prefix property.

It is interesting to see whether the axioms of system $F_1$ (presented
on page 7) are still valid if we transform both sides using $\Pi$ with
$\Epsi$ as the continuation PEG, that is, if $l$ is the left side of
the equation and $r$ is the right side then $\Mat{\Pi(l,\Epsi)}{xy}
\Lp y$ if and only if $\Mat{\Pi(r,\Epsi)}{xy} \Lp y$. This is
straightforward for axioms 1, 2, 4, 6, 7, 8, and 9; in fact, it is
easy to prove that these axioms will be valid for any PEG, not just
PEGs obtained from our transformation.

Applying $\Pi$ to both sides of axiom 2, $e_1(e_2e_3)$ and
$(e_1e_2)e_3$, makes them identical; they both become
$\Pi(e_1,\Pi(e_2,\Pi(e_3,G_k)))$. The same thing happens with axiom 5,
$(\Chre{e_1}{e_2})e_3 = \Chre{e_1e_3}{e_2e_3}$; the transformation of
the left side, $\Pi((\Chre{e_1}{e_2})e_3, G_k)$, becomes
$\Chre{\Pi(e_1,\Pi(e_3,G_k))}{\Pi(e_2,\Pi(e_3,G_k))}$ via
the intermediate expression $\Pi(\Chre{e_1}{e_2},\Pi(e_3,G_k))$, while the transformation of the
right side, $\Pi(\Chre{e_1e_3}{e_2e_3},G_k)$, also becomes
$\Chre{\Pi(e_1,\Pi(e_3,G_k))}{\Pi(e_2,\Pi(e_3,G_k))}$, although via
the intermediate expression $\Chre{\Pi(e_1e_3,G_k)}{\Pi(e_2e_3,G_k)}$.

The transformation of axiom 3, $\Chre{e_1}{e_2} = \Chre{e_2}{e_1}$,
will not be valid; the left side becomes the PEG
$\Chre{\Pi(e_1,G_k)}{\Pi(e_2,G_k)}$ and the right side becomes the PEG
$\Chre{\Pi(e_2,G_k)}{\Pi(e_1,G_k)}$, but ordered choice is not
commutative in the general case. One case where this choice is
commutative is if the language of $\Chre{e_1}{e_2}$ has the prefix
property. We can use an argument analogous to the argument of
Lemma~\ref{lemma:repeglang} to prove this, which is not surprising, as
this lemma together with Lemma~\ref{lemma:salomaa} implies that this
axiom should hold for expressions with languages that have the prefix
property.

Axiom 10, $e^* = \Chre{\Epsi}{e^* e}$, needs to be rewritten as $e^* =
\Chre{e^* e}{\Epsi}$ or it is trivially not valid, as the right side
will always match just $\Epsi$. Again, this is not surprising, as
$\Chre{\Epsi}{e^* e}$ does not have the prefix property, and this is
the same behavior of regex implementations. Rewriting the axiom as $e^* =
\Chre{e^* e}{\Epsi}$ makes it valid when we apply $\Pi$ to both
sides, as long as $e^*$ is well-formed.

 The left side becomes the PEG $A$ where $A \rightarrow
\Chre{\Pi(e,A)}{\Epsi}$, while the right side becomes $\Chre{B}{\Epsi}$
where $B \rightarrow \Chre{\Pi(e,B)}{\Pi(e,\Epsi)}$. If $\Pi(e,A)$
fails it means that $\Pi(e,\Epsi)$ would also fail, and so do
$\Pi(e,B)$, then $B$ fails and the $\Chre{B}{\Epsi}$ succeeds by {\bf
  choice.2}. Analogous reasoning holds for the other side, if $B$
fails. If $\Pi(e,A)$ succeeds then $\Pi(e,\Epsi)$ matches a non-empty
prefix, and $A$ matches the rest, and we can assume that
$\Chre{B}{\Epsi}$ matches this rest by induction on the length of the
matched string. We can use this to conclude that $\Chre{B}{\Epsi}$
also succeeds. Again, analogous reasoning holds for the converse.

The right side of axiom 11, $(\Chre{\Epsi}{e})^*$, is not well-formed,
and applying $\Pi$ to it would lead to a left-recursive PEG with no
possible proof tree. We still need to show that any regular expression can be made
well-formed without changing its language. This is the topic of the
next section, where we give a transformation that rewrites
non-well-formed repetitions so they are well-formed with minimal
changes to the structure of the original regular expression. Applying
this transformation to the right side of axiom 11 will make it
identical to the left side, making the axiom trivially valid.

\subsection{Transformation of Repetitions $\E^*$ where $\Epsi \in L(\E)$}
\label{sec:left}

A regular expression $\E$ that has a subexpression $e_i^*$ where $e_i$
can match the empty string is not well-formed. As $e_i$ can succeed
without consuming any input one outcome of $e_i^*$ is to stay in the
same place of the input indefinitely. Regex libraries that rely on
backtracking may enter an infinite loop with non-well-formed
expressions unless they take measures to avoid it, using ad-hoc
rules to detect and break the resulting infinite loops~\cite{perl:emptyrep}.

When $\E$ is not well-formed, the PEG we obtain through transformation $\Pi$ is not complete. A PEG that is not complete can make a PEG library enter an infinite loop. To show an example on how a non-well-formed regular expression leads to a PEG that is not complete, let us transform $\Con{(\Chre{a}{\Epsi})^*}{b}$ using $\Pi$. Using $\Epsi$ as continuation yields the following PEG $A$:
\[
A \,\arrow\,
\Choexe{\Con{a}{A}}{\Choexe{A}{b}}
\]

The PEG above is {\em left recursive}, so it is not complete. In fact, this PEG does not have a proof tree for any input, so it is not equivalent to the regular expression $\Con{(\Chre{a}{\Epsi})^*}{b}$.

Transformation $\Pi$ is not correct for non-well-formed regular expressions, but we can make any non-well-formed regular expression well-formed by rewriting repetitions $e_i^*$ where $\Epsi \in L(e_i)$ as ${e_i^\prime}^*$ where $\Epsi \not\in L(e_i^\prime)$ and $L({e_i^\prime}^*) = L(e_i^*)$. The regular expression above would become $\Chre{a^*}{b}$, which $\Pi$ transforms into an equivalent complete PEG. 

This section presents a transformation that mechanically performs this
rewriting. We use a pair of functions to rewrite an expression,
$\Fout$ and $\Fin$. Function $\Fout$ recursively searches for a
repetition that has $\Epsi$ in the language of its subexpression,
while $\Fin$ rewrites the repetition's subexpression so it is
well-formed, does not have $\Epsi$ in its language, and does not
change the language of the repetition. Both $\Fin$ and $\Fout$ use two
auxiliary predicates, $\Isnull$ and $\Hase$, that respectively test if
an expression is equal to $\Epsi$ (if its language is the singleton
set $\{ \Epsi \}$) and if an expression has $\Epsi$ in
its language. Figure~\ref{fig:isnull} has inductive definitions for the $\Isnull$ and $\Hase$ predicates. 

\begin{figure}[t]
{\small
\begin{eqnarray*}
\Isnull(\Epsi) & = & \True \\
 \Isnull(a) & = & \False \\
 \Isnull(e^*) &  = & \Isnull(e)  \\
\Isnull(\Con{e_1}{e_2}) & =  &\Land{\Isnull(e_1)}{\Isnull(e_2)} \\
\Isnull(\Choice{e_1}{e_2}) & = & \Land{\Isnull(e_1)}{\Isnull(e_2)} \\
& & \\
\Hase(\Epsi) & = & \True \\
 \Hase(a) & = & \False \\
 \Hase(e^*) & = &  \True \\
\Hase(\Con{e_1}{e_2}) & = & \Land{\Hase(e_1)}{\Hase(e_2)} \\
\Hase(\Choice{e_1}{e_2}) & = & \Lor{\Hase(e_1)}{\Hase(e_2)}
\end{eqnarray*}
\caption{Definition of predicates $\Isnull$ and $\Hase$}
\label{fig:isnull}
}
\end{figure}

\begin{figure}[t]
\begin{align*}
\Fout(\E) & \Myeq \E \mbox{, if } e = \Epsi \mbox{ or } e = a \\
\Fout(\Con{e_1}{e_2}) & \Myeq
    \Con{\Fout(e_1)}{\Fout(e_2)} \\
\Fout(\Choice{e_1}{e_2}) & \Myeq
    \Choexe{\Fout(e_1)}{\Fout(e_2)} \\ 
\Fout(e^*) & \Myeq
\left\{ \begin{array}{ll}
	\Fout(e)^* & \mbox{if\, $\Nemp(e)$} \\
	\Epsi      & \mbox{if\, $\Isnull(e)$} \\
	\Fin(e)^* &  \mbox{otherwise}
\end{array}
\right.
\end{align*}
\caption{Definition of Function $\Fout$}
\label{fig:fout}
\end{figure}

Function $\Fout$ is simple: for the base expressions it is the
identity, for the composite expressions $\Fout$ applies itself
recursively to subexpressions unless the expression is a repetition
where the repetition's subexpression matches $\Epsi$. In this case
$\Fout$ transforms the repetition to $\Epsi$ if the subexpression is equal to $\Epsi$ (as $\Epsi^* \equiv \Epsi$), or uses $\Fin$ to rewrite the subexpression. Figure~\ref{fig:fout} has the inductive definition of $\Fout$. It obeys the following lemma:

\begin{lemma}
If $\Fin(\Ek)$ is well-formed, $\Epsi \not\in L(\Fin(\Ek))$, and $L(\Fin(\Ek)^*) = L(\Ek^*)$ for any $\Ek$ with $\Epsi \in L(\Ek)$ and $L(\Ek) \neq \Epsi$ then, for any $\E$, $\Fout(\E)$ is well-formed and $L(\E) = L(\Fout(\E))$.
\end{lemma}

\begin{proof}
By structural induction on $\E$. Inductive cases follow directly from the induction hypothesis, except for $\E^*$ where $\Epsi \in L(\E)$, where it follows from the properties of $\Fin$.
\end{proof}

\begin{figure}[t]
\begin{align*}
\Fin(\Con{e_1}{e_2}) & \Myeq
    \Fin(\Choice{e_1}{e_2}) \\
\Fin(\Choice{e_1}{e_2}) & \Myeq
%%% begin Choice %%%
\left\{ \begin{array}{ll}
  \Fin(e_2)
    & \mbox{if\, $\Isnull(e_1)$\,
	  and\, $\Hase(e_2)$} \\
  \Fout(e_2)
    & \mbox{if\, $\Isnull(e_1)$\,
		and\, $\Nemp(e_2)$} \\
  \Fin(e_1)
    & \mbox{if\, $\Hase(e_1)$\,
    and\, $\Isnull(e_2)$} \\
  \Fout(e_1)
    & \mbox{if\, $\Nemp(e_1)$\,
    and\, $\Isnull(e_2)$} \\
  \Choexe{\Fout(e_1)}{\Fin(e_2)}
    & \mbox{if\, $\Nemp(e_1)$\,
		and\, $\Nnull(e_2)$} \\
  \Choexe{\Fin(e_1)}{\Fout(e_2)}
    & \mbox{if\, $\Nnull(e_1)$\,
		and\, $\Nemp(e_2)$} \\
  \Choexe{\Fin(e_1)}{\Fin(e_2)}
    & \mbox{otherwise}
\end{array}
\right. \\
%%% end Choice %%%
\Fin(e^*) & \Myeq
\left\{ \begin{array}{ll}
	\Fin(e)
    & \mbox{if\, $\Hase(e)$} \\
  \Fout(e)
		& \mbox{otherwise}
\end{array}
\right.
\end{align*}
\caption{Definition of Function $\Fin(\E)$,
where $\Nnull(\E)$ \,and\, $\Hase(\E)$}
\label{fig:fin}
\end{figure}

Function $\Fin$ does the heavy lifting of the rewriting, it is used when $\Fout$ finds an expression $\E^*$ 
where $\Nnull(\E)$ \,and\, $\Hase(\E)$.
 If its argument is a repetition it throws away the repetition because
 it is superfluous. Then $\Fin$ applies $\Fout$ or itself to the
 subexpression depending on whether it matches $\Epsi$ or not. If the
 argument is a choice $\Fin$ throws away one of the sides if its equal
 to $\Epsi$, as it is superfluous because of the repetition, and
 rewrites the remaining side using $\Fout$ or $\Fin$ depending on
 whether it matches $\Epsi$ or not. In case both sides are not equal
 to $\Epsi$ $\Fin$ rewrites both. If the argument is a concatenation
 $\Fin$ rewrites it as a choice and applies itself to the
 choice. 

Transforming a concatenation into a choice obviously is not a
 valid transformation in the general case, but it is safe in the
 context of $\Fin$; $\Fin$ is working inside a repetition
 expression, and its argument has $\Epsi$ in its language, so we can
 use an identity involving languages and the Kleene
 closure that says $(AB)^* = (A \cup B)^*$ if $\Epsi \in A$ and $\Epsi \in B$.
Figure~\ref{fig:fin} has the inductive definition of $\Fin$. It obeys the following lemma:

\begin{lemma}
If $\Fout(\Ek)$ is well-formed and $L(\Fout(\Ek)) = L(\Ek)$ for any $\Ek$ then, for any $\E$ with $\Epsi \in L(\E)$ and $L(\E) \neq \{\Epsi\}$, $\Epsi \notin L(\Fin(\E))$, $L(\E^*) = L(\Fin(\E)^*)$, and $\Fin(\E)$ is well-formed.
\end{lemma}

\begin{proof}
By structural induction on $\E$. Most cases follow directly from the induction hypothesis and the properties of $\Fout$. The subcases of choice where the result is also a choice use the Kleene closure property $(A \cup B)^* = (A^* \cup B^*)^*$ together with the induction hypothesis and the properties of $\Fout$. Concatenation becomes to a choice using the property mentioned above this lemma.
\end{proof}

As an example, let us use $\Fout$ \,and\, $\Fin$ 
to rewrite the regular expression
\,$(\Chre{\Con{b}{c}}{\Con{a^*}{(\Chre{d}{\Epsi})}})^*$\,
into a well-formed regular expression. We show the sequence of steps below:
\begin{align*}
\Fout((\Chre{\Con{b}{c}}{\Con{a^*}{(\Chre{d}{\Epsi})}})^*) & \Myeq 
(\Fin(\Chre{\Con{b}{c}}{\Con{a^*}{(\Chre{d}{\Epsi})}}))^* 
 \Myeq
(\Chre{\Fout(\Con{b}{c})}{\Fin(\Con{a^*}{(\Chre{d}{\Epsi})})})^* \\ 
& \Myeq
(\Chre{\Con{\Fout(b)}{\Fout(c)}}{(\Chre{\Fin(a^*)}{\Fin(\Chre{d}{\Epsi})})})^* \\
& \Myeq 
(\Chre{\Con{b}{c}}{(\Chre{\Fout(a)}{\Fout(d)})})^*
\Myeq
(\Chre{\Con{b}{c}}{(\Chre{a}{d})})^*
\end{align*}

The idea is for rewriting to be automated, and transparent to the user
of regex libraries based on our transformation, unless the user wants
to see how their expression can be simplified. Notice that just the
presence of $\Epsi$ inside a repetition does not mean that a regular
expression is not well-formed. The
\,$(\Chre{\Con{b}{c}}{\Con{a}{(\Chre{d}{\Epsi})}})^*$\, expression
looks very similar to the previous one, but is well-formed and
left unmodified by $\Fout$.

\section{Optimizing Search and Repetition}
\label{sec:opt}

A common application of regexes is to search for parts of a subject
that match some pattern, but our formal model of regular expressions
and PEGs is {\em anchored}, as our matches must start on the first
symbol of the subject instead of starting anywhere. It is easy to
build a PEG that will search for another PEG \Grm{V}{T}{P}{S}, though,
we just need to add a new non-terminal $S^\prime$ as starting pattern,
with $S^\prime \rightarrow \Choice{S}{\Pdot S^\prime}$ as a new
production, where $\Pdot$ is a shortcut for a regular expression that
matches any terminal. If trying to match $S$ from the beginning of the
subject fails, then the PEG skips the first symbol and tries again on
the second.

The search pattern works, but can be very inefficient if the PEG
engine always has to use backtracking to implement ordered choice, as
advancing to the correct starting position may involve a large amount
of advancing and then backtracking. A related problem occurs when
converting regex repetition into PEGs, as the PEG generated from the
regular expression $e_1^* e_2$ will greedily try to consume as much of
the subject with $e_1$ as possible, then try $e_2$ and backtrack each
match of $e_1$ until $e_2$ succeeds or the whole pattern fails. In the
rest of this section we will show how we can use properties of the
expressions we are trying to search or match in conjunction with
syntactic predicates to reduce the amount of backtracking necessary in
both searches and repetition expressions.

\subsection{Search}

The search pattern for a PEG tries to match the PEG then advances one position in the subject and tries again if the match fails. A simple way to improve this is to advance several positions if possible, skipping starting positions that have no chance of a successful match. If we know that a successful match always consumes part of the subject and begins with a symbol in a set $\mathcal{F}$ then we can skip a failing starting position with the pattern $![\mathcal{F}] .$, where $[\mathcal{F}]$ is a $character set$ pattern that matches any symbol in the set. We can skip a string of failing symbols with the pattern $(![\mathcal{F}]\ \Pdot)^*$. The new search expression for the PEG with starting pattern $p$ can be written as follows:
\[
S \rightarrow (![\mathcal{F}]\ \Pdot)^* (\Choice{p}{\Pdot S})
\]

The set $\mathcal{F}$ for a pattern $p$ derived from a regular expression $e$ is just the $\Fst$ set of $e$, which has a simple definition in terms of the $\Lr$ relation below:
\[
\mathit{FIRST}(e) = \{ a \in T \, | \, \exists x,y \ e \, axy \Lr y, \
x, y \in T^*\}
\]

It is easy to prove that the two search expressions are equivalent by induction on the height of the corresponding derivation trees. The tricky case, where $(![\mathcal{F}]\ \Pdot)^*$, just uses the definition of $\Fst$ to build a tree of successive applications of rule {\bf ord.2} until we can use the induction hypothesis in its rightmost leaf.

\subsection{Repetition}

If two regular expressions $e_1$ and $e_2$ have disjoint $\Fst$ sets then it is safe to match $e_1^* e_2$ using possessive repetition. This means that we can transform $e_1^* e_2$ into the PEG $p_1^* p_2$ where $p_1$ and $p_2$ are the PEGs we get from transforming $e_1$ and $e_2$. Formally, we can define $\Rp{e_1^* e_2}{\Gk}$ when $\Fst(e_1) \cap \Fst(e_2) = \emptyset$ as follows, where we use $G_2[\varepsilon]$ as the continuation for transforming $e_1$ just to avoid collisions on the names of non-terminals:
\begin{align*}
 \Rp{e_1^* e_2}{\Gk}  & =  \Grm{V_1}{T}{P_1}{p_1^* p_2} \\
 \mbox{where }  \Grm{V_1}{T}{P_1}{p_1}  & =  \Rp{e_1}{G_2[\varepsilon]} \\
    \mbox{and } G_2 \;=\; \Grm{V_2}{T}{P_2}{p_2} & =  \Rp{e_2}{\Gk}
\end{align*}

The easiest way to prove the correctness of the new rule is by proving the equivalence of the PEGs we get from $\Rp{e_1^* e_2}{\Gk}$ using the old and new rule. This is a straightforward induction on the height of the proof trees for these PEGs, using the fact that disjoint $\Fst$ sets for $e_1$ and $e_2$ implies disjointedness of their equivalent PEGs.

In the general case, where the $\Fst$ sets of $e_1$ and $e_2$ are not disjoint, we can still avoid some amount of backtracking on $e_1^* e_2$ by being possessive whenever there is no chance of $e_2$ doing a successful match, as backtracking to a point where $e_2$ cannot match is useless. The idea is to use a predicated repetition of $p_1$ before doing the choice $\Choice{p_1 A}{p_2}$ that guarantees that the PEG will backtrack to a point where $p_2$ matches, if possible. We can use the $\Fst$ set of $e_2$ as an approximation to the set of possible matches of $e_2$, and the PEG for $e_1^* e_2$ becomes $A \rightarrow (![\Fst(e_2)]\ p_1)^* (\Choice{p_1 A}{p_2})$. The full rule for $\Rp{e_1^* e_2}{\Gk}$  becomes as follows:
\begin{align*}
 \Rp{e_1^* e_2}{\Gk}  =  \Grm{V_1 \cup \{A\}}{T}{P_1 \cup \{A \rightarrow (![\Fst(e_2)]\ p_1)^* (\Choice{p_1 A}{p_2})\}}{A} \\
 \mbox{where }  \Grm{V_1}{T}{P_1}{p_1}   = \Rp{e_1}{G_2[\varepsilon]} \\
\mbox{with } A \notin V_1 \mbox{ and }  G_2  = \Grm{V_2}{T}{P_2}{p_2} = 
                            \Rp{e_2}{\Gk}
\end{align*}

Again, the easiest way to prove that this new rule is correct is by proving the equivalence of the PEGs obtained from the old and the new rule, by induction on the height of the corresponding proof trees.

\subsection{Combining Search and Repetition}

We can further optimize the case where we want to search for the pattern $e_1^* e_2$ or $e_1 e_1^* e_2$ (we will use $e_1^+$ as a shorthand for $e_1 e_1^*$), and all strings in $L(e_1)$ have length one. We can safely skip the prefix of the subject that matches a possessive repetition of $e_1$ before trying again, because if the pattern would match from any of these positions then it would not have failed in the first place. We can combine this with our first search optimization to yield the following search pattern:
\[
S \rightarrow (![\mathcal{F}]\ \Pdot)^* (\Choice{p}{p_1^* S})
\]

In the pattern above, $p$ is the starting expression of a PEG equivalent to the regular expression we are searching and $p_1$ is the starting expression of a PEG equivalent to $e_1$ with an empty continuation. Set $\mathcal{F}$ is still the $\Fst$ set of the whole regular expression. If the $\Fst$ sets of $e_1$ and $e_2$ are disjoint we can further optimize our search by breaking up $p$ and using the following search expression to search for $e_1^* e_2$:
\[
S \rightarrow (![\mathcal{F}]\ \Pdot)^* p_1^* (\Choice{p_2}{S})
\]

The special case searching for $e_1^+ e_2$ just uses the search expression $S \rightarrow (![\mathcal{F}]\ \Pdot)^* p_1^+ (\Choice{p_2}{S})$. Proofs that these optimizations are correct are straightforward, by proving that these search expressions are equivalent to $S \rightarrow (![\mathcal{F}]\ \Pdot)^* (\Choice{p}{\Pdot S})$ by induction on the height of the derivation trees.

\section{Benchmarks}
\label{sec:bench}

This section presents some benchmarks that compare a regex engine
based on an implementation of our transformation with the resulting
PEGs executed with LPEG, a fast backtracking PEG
engine~\cite{roberto:lpeg}. We compare this engine with PCRE, a backtracking
regex engine that performs ad-hoc optimizations to cut the amount of
backtracking needed~\cite{pcre}, and with RE2, a non-backtracking
(automata-based) regex engine that nevertheless also incorporates
ad-hoc optimizations~\cite{regwild}.

We tested our search and repetition optimizations with a series of
benchmarks that search for the first successful match of a regular
expression inside a large subject, the Project Gutenberg version of
the King James Bible~\cite{bible}. Our first benchmark searches for a
single literal word in the subject, and serves as a simple test of the
search optimization. Table~\ref{tab:searchword} shows the results. We
can see that the optimization is very effective, as LPEG optimizes the
repetition in the search pattern to a single instruction of its
parsing machine that scans the subject checking each character against
a bitmap that encodes the character set. RE2 and PCRE use ad-hoc
optimizations to find the string and are still faster in some of the
cases~\cite{regwild}.

\begin{table}
{\small
\centering
\begin{tabular}{lrrrrr} \hline
Word                     & RE2 & PCRE & Unoptimized & Search &  Line  \\ \hline
\texttt{Geshurites}      &   1 &    1 &   12        &   1  &  19936 \\ 
\texttt{worshippeth}     &   3 &    3 &   25        &   4  &  42140 \\ 
\texttt{blotteth}        &   3 &    3 &   33        &   6  &  60005 \\ 
\texttt{sprang}          &   7 &    9 &   47        &  11  &  80000 \\ 
%\texttt{Even so, come, Lord Jesus}  
 %                        &   4 &  55  &   7  &  99811 \\ 
\hline 
\end{tabular}
\caption{Time in milliseconds to search for a word}
\label{tab:searchword}
}
\end{table}

Our second benchmark searches for two literal words in the same period (separated by letters, spaces or commas), and we test the search and repetition optimizations, but cannot apply the combined optimization of Section~4.3 because the expression does not have the necessary structure. Table~\ref{tab:twowords1order} shows the results, and we separate the optimizations to show the contribution of each one in the final result. The runtime is still dominated by having to find where in the subject the match is, so optimizing the repetition inside the pattern does not yield any gains. The pattern starts with a literal, so RE2 and PCRE are using ad-hoc optimizations to find where the match begins.

\begin{table}
{\small
\centering
\begin{tabular}{lrrrrrr} \hline
Words                      & RE2  & PCRE & Unopt & Search & Repetition & Line  \\ \hline
\texttt{Adam - Eve}        &   1  &   0  &    0  &     0  &       0    &   261 \\ 
\texttt{Israel - Samaria}  &   2  &   2  &   32  &     3  &       3    & 31144 \\ 
\texttt{Jesus - John}      &   2  &   3  &   73  &     6  &       6    & 76781 \\ 
\texttt{Jesus - Judas}     &   2  &   4  &   81  &     6  &       6    & 84614 \\ 
\texttt{Jude - Jesus}      &   3  &   4  &   94  &     7  &       7    & 98311 \\  
\texttt{Abraham - Jesus}   &   5  &   5  &   96  &     8  &       8    & no match  \\ \hline
%\texttt{the - Jesus}       &   5  &   5  &  119  &    57  &      25    & 77263  \\ \hline 
\end{tabular}
\caption{Time in milliseconds to search for two words in the same period}
\label{tab:twowords1order}
}
\end{table}

The third benchmark searches for a literal word that follows any other word plus a single space (a regular expression $[a-zA-Z]^+\mbox{\textvisiblespace} w$, using character class notation and \textvisiblespace\ for the empty space symbol). This pattern falls in the case where the $\Fst$ sets of the repeated pattern and the pattern following the repetition are disjoint. We can apply the combined search and repetition optimization for this pattern, and compare it with the basic search and repetition optimizations. Table~\ref{tab:reppos} shows the results. Now even the unoptimized PEG defeats a backtracking regex matcher, but the DFA-based RE2 is much faster. The $\Fst$ set of the pattern includes most of terminals, and the search optimization is not effective. The biggest gain comes from the combined optimization, as it lets the PEG skip large portions of the subject in case of failure, yielding a result that is much closer to RE2.

\begin{table}
{\small
\centering
\begin{tabular}{lrrrrrrr} \hline
Word                     & RE2 & PCRE & Unopt & Search & Rep & Combined & Line  \\ \hline
\texttt{Geshurites}      &   5 &  74  &   57  &   59   &  38 &    8     &  19995 \\ 
\texttt{worshippeth}     &   7 & 156  &  121  &  126   &  86 &   18     &  42140 \\ 
\texttt{blotteth}        &  10 & 208  &  159  &  167   & 121 &   24     &  60005 \\ 
\texttt{sprang}          &  12 & 285  &  222  &  227   & 147 &   32     &  80000 \\ 
%\texttt{Even so, come, Lord Jesus}  
 %                                  &  16 &  289  & 226   &  76  &  39  & no match  \\ 
\hline 
\end{tabular}
\caption{Time in milliseconds to search for a word following another}
\label{tab:reppos}
}
\end{table}

The fourth and final benchmark extends the second benchmark by bracketing the pattern with a pattern that matches the part of the period that precedes and follows the two words we are searching, yielding the pattern $[a-zA-Z,\mbox{\textvisiblespace}]^*w_1[a-zA-Z,\mbox{\textvisiblespace}]^*w_2[a-zA-Z,\mbox{\textvisiblespace}]^*$. There is overlap in the $\Fst$ sets of $[a-zA-Z,\mbox{\textvisiblespace}]$, $w_1$, and $w_2$, so we need to use the more general form of the repetition optimization. We can also apply the combined optimization. As in the third benchmark, we compare this optimization with the basic search and repetition optimizations. Table~\ref{tab:semipos1} shows the results. The effect of the repetition optimization is bigger in this benchmark, but what brings the performance close to a DFA-based regex matcher, and much better than a backtracking regex matcher, is still the combined optimization.

\begin{table}
{\small
\centering
\begin{tabular}{lrrrrrrr} \hline
Words                        & RE2 & PCRE & Unopt & Search & Rep & Comb &  Line  \\ \hline
\texttt{Adam - Eve}          &   2 &    4 &    6  &    6   &   0 &    0     &    261 \\ 
\texttt{Israel - Samaria}    &   6 &  504 &  752  &  750   & 126 &    8     &  31144 \\ 
\texttt{Jesus - John}        &  12 & 1134 & 1710  & 1718   & 278 &   18     &  76781 \\ 
\texttt{Jesus - Judas}       &  13 & 1246 & 1884  & 1892   & 306 &   20     &  84614 \\ 
\texttt{Jude - Jesus}        &  15 & 1446 & 2176  & 2188   & 364 &   24     &  98311 \\  
\texttt{Abraham - Jesus}     &  15 & 1470 & 2214  & 2220   & 362 &   24     &  no match \\ \hline 
\end{tabular}
\caption{Time in milliseconds to search for a period containing two words}
\label{tab:semipos1}
}
\end{table}

Our benchmarks show that without optimizations our PEG-based engine
performs on par with PCRE on more complex patterns. The optimizations
bring it to a factor of 1 to 3 of the performance of RE2, a very
efficient and well-tuned regex implementation that cannot implement
common regex extensions due to its automata-based implementation approach.

\section{Transforming Regex Extensions}
\label{sec:exts}

Regexes add several ad-hoc extensions to regular expressions. We can
easily adapt transformation $\Pi$ to deal with some of these
extensions, and this section shows how to use $\Pi$ with independent expressions, possessive repetitions, lazy repetitions, and lookahead. An informal but broader discussion of regex extensions in the context of translation to PEGs was published by Oikawa
et al.~\cite{japa:sblp}.

\begin{figure}[t]
{%\small
\begin{align*}
\Rp{\Eind}{\Gk}  & \Myeq   \Grm{V_1}{T}{P_1}{\Con{p_1}{\Pk}} 
		 \mbox{, where} \;\,
\Grm{V_1}{T}{P_1}{p_1} \Myeq \Rp{e_1}{\Gk[\Epsi]} \\ 
\Rp{\Epos}{\Gk} & \Myeq  \Rp{?\rangle e^*}{\Gk} \\ 
\Rp{\Elazy}{\Gk} & \Myeq G \;
		 \mbox{, where} \;\,
G \Myeq \Grm{V_1}{T}{P_1 \cup \{A \arrow \Choice{\Pk}{p_1}\}}{A} \;\mbox{,}\\
&		 \;\;\;\,  
     \Grm{V_1}{T}{P_1}{p_1} \Myeq 
                            \Rp{e_1}{\Grm{\Vk \cup \{A\}}{T}{P_k}{A}}, \;\mbox{and}\;\; A \notin \Vk \\
\Rp{\Enot}{\Gk}  & \Myeq   \Grm{V_1}{T}{P_1}{\Con{!p_1}{\Pk}} 
		 \mbox{, where} \;\,
\Grm{V_1}{T}{P_1}{p_1} \Myeq \Rp{e_1}{\Gk[\Epsi]} \\
\Rp{\Eand}{\Gk}  & \Myeq   \Grm{V_1}{T}{P_1}{\Con{!!p_1}{\Pk}} 
		 \mbox{, where} \;\,
\Grm{V_1}{T}{P_1}{p_1} \Myeq \Rp{e_1}{\Gk[\Epsi]} 
\end{align*}
\caption{Adapting Function $\Pi$ to Deal with Regex Extensions}
\label{fig:extpi}
}
\end{figure}

The regex $\Eind$ is an independent expression (also known as atomic
grouping). It matches independently of the expression that follows it,
so a failure when matching the expression that follows $\Eind$ does
not force a backtracking regex matcher to backtrack to $\Eind$'s
alternative matches. This is the same behavior as a PEG, so to
transform $\Eind$ we first transform it using an empty continuation,
then concatenate the result with the original continuation.

The regex $\Epos$ is a possessive repetition. It always matches as
most as possible of the input, even if this leads to a subsequent
failure. It is the same as $?\rangle e^*$ if the longest-match rule is
used. The semantics of $\Pi$ guarantees longest match, so it uses this
identity to transform $\Epos$.

The regex $\Elazy$ is a lazy repetition. It always matches as little
of the input as necessary for the rest of the expression to match
(shortest match). The transformation of this regex is very similar to
the transformation of $e_1^*$, we just flip $p_1$ and $p_k$ in the
production of non-terminal $A$. Now the PEG tries to match the rest of
the expression first, and will only try another step of the repetition
if the rest fails.

The regex $\Enot$ is a negative lookahead. The regex matcher tries to
match the subexpression; it it fails then the negative lookahead
succeeds without consuming any input, and if the subexpression
succeeds the negative lookahead fails. Negative lookahead is also an
independent expression. Transforming this regex is just a matter of
using PEGs negative lookahead, which works in the same way, on the
result of transforming the subexpression as an independent expression.

Finally, the regex $\Eand$ is a positive lookahead, where the regex
matcher tries to match the subexpression and fails if the
subexpression fails and succeeds if the subexpression succeeds, but
does not consume any input. It is also an independent expression. We
transform a positive lookahead by transforming the subexpression as an
independent expression and then using PEGs negative lookahead twice.

None of these extensions has been formalized before, as they depend on the
behavior of backtracking-based implementations of regexes instead of
the semantics of regular expressions. We decided to formalize them in
terms of their conversion to PEGs instead of trying to rework our
semantics of regular expressions to accommodate them, as these
extensions map naturally to concepts that are already part of the
semantics of PEGs.

\begin{figure}[t]
\begin{align*}
\Fout(\Eind) & \Myeq ?\rangle \Fout(e_1) \\
\Fout(\Epos) & \Myeq
\left\{ \begin{array}{ll}
	\Eposs{\Fout(e_1)} & \mbox{if\, $\Nemp(e_1)$} \\
	\Epsi      & \mbox{if\, $\Isnull(e_1)$} \\
	\Eposs{\Fin(e_1)} &  \mbox{otherwise}
\end{array}
\right. \\
\Fout(\Elazy) & \Myeq
\left\{ \begin{array}{ll}
	\Elazyy{\Fout(e_1)} & \mbox{if\, $\Nemp(e_1)$} \\
	\Epsi      & \mbox{if\, $\Isnull(e_1)$} \\
	\Elazyy{\Fin(e_1)} &  \mbox{otherwise}
\end{array}
\right. \\
\Fout(\Enot) & \Myeq \Enott{\Fout(e_1)} \\
\Fout(\Eand) & \Myeq \Eandd{\Fout(e_1)} \\
& \\
\Fin(\Eind) & \Myeq \Eindd{\Fin(e_1)} \\
\Fin(\Epos) & \Myeq
\left\{ \begin{array}{ll}
	\Fin(e_1)
    & \mbox{if\, $\Hase(e_1)$} \\
  \Fout(e_1)
		& \mbox{otherwise}
\end{array}
\right. \\
\Fin(\Elazy) & \Myeq
\left\{ \begin{array}{ll}
	\Fin(e_1)
    & \mbox{if\, $\Hase(e_1)$} \\
  \Fout(e_1)
		& \mbox{otherwise}
\end{array}
\right.
\end{align*}
\caption{Definition of Functions $\Fout$ and $\Fin$ for regex extensions}
\label{fig:foutfin}
\end{figure}

The well-formedness rewriting of Section~\ref{sec:left} needs to
accommodate the new extensions. It is not possible to rewrite all
non-well-formed expressions with these extensions while keeping their
behavior the same, as these extensions make it possible to write
expressions that cannot give a meaningful result,
such as $(\Eindd{(\Chre{\Epsi}{a})})^*$ or
$(\Eandd{a}(\Chre{d}{\Epsi}))^*$. Other expressions can work with some
subjects and not work with others, such as $(\Eindd{(\Chre{a}{\Chre{\Epsi}{b}})})^*$ or
$(\Eandd{a}(\Chre{a}{\Epsi}))^*$.

Our approach will be to rewrite problematic expressions so they give
the same result for the subjects where they do not cause problems, but
also give a result for other subjects, that is, they will match a
superset of the strings that the original expression matches. For example, the four
expressions above will be respectively rewritten to $(\Eindd{a})^*$,
$d^*$, $(\Eindd(\Chre{a}{b}))^*$, and $a^*$.

The $\Isnull$ predicate is $\True$ for $\Enot$ and $\Eand$
expressions,  and $\Isnull(e_1)$ for the other extensions. This is a conservative
definition, as expressions such as $?\rangle (\Chre{\Epsi}{e})$ can
also be replaced by $\Epsi$ given the informal semantics of
regexes. The $\Hase$ predicate is $\True$ for all the extensions except
$\Eind$, where it is $\Hase(e_1)$.

Figure~\ref{fig:foutfin} gives the definitions of $\Fout$ and $\Fin$
for the extensions. Function $\Fout$ just applies itself recursively
for $\Eind$, $\Enot$ and $\Eand$, but it needs to rewrite the
repetitions using $\Fin$ if their bodies can match $\Epsi$. Function
$\Fin$ applies itself recursively to atomic groupings, keeping them
atomic, but it strips repetitions. A repetition being rewritten by
$\Fin$ is used directly inside another repetition, so it does not
matter if it is possessive, lazy, or a regular greedy repetition,  it
is the outer repetition that will govern how much of the subject will
be matched. 

We do not need to define $\Fin$ for negative and positive
lookaheads, as pathological uses of these expressions are eliminated
by the $\Fout$ case that rewrites repetitions with an $\Isnull$ body and
the $\Fin$ cases that rewrite choices with $\Isnull$ alternatives.

The extensions do not impact the optimizations of Section~\ref{sec:opt} if we
provide a way of computing a $\Fst$ set for them, as the optimizations
do not depend on the structure of the subexpressions they use. We obviously
cannot apply the repetition optimization on $\Epos e_2$, $\Elazy
e_2$, or $\Eindd{(e_1^*)}\, e_2$, but applying it on $e_1^*e_2$ where
extensions appear inside $e_1$ or $e_2$ is not a problem. The
repetition optimizations are turning repetitions into possessive
repetitions where possible, so not being able to optimize expressions
such as the ones above is not a loss, as they will already exhibit
good backtracking behavior.

\begin{figure}[t]
\begin{align*}
\Fst(\Epsi) & \Myeq \emptyset \\
\Fst(a) & \Myeq \{ \mathtt{a} \} \\
\Fst(\Con{e_1}{e_2}) & \Myeq 
\left\{ \begin{array}{ll}
	\Fst(e_1) & \mbox{if\, $\Nemp(e_1)$} \\
	\Fst(e_1) \cup \Fst(e_2)   & \mbox{if\, $\Hase(e_1)$}
\end{array}
\right. \\
\Fst(\Chre{e_1}{e_2}) & \Myeq \Fst(e_1) \cup \Fst(e_2) \\
\Fst(e^*) & \Myeq \Fst(e) \\
\Fst(\Eindd{e}) & \Myeq \Fst(e) \\
\Fst(\Eposs{e}) & \Myeq \Fst(e) \\
\Fst(\Elazyy{e}) & \Myeq \Fst(e) \\
\Fst(\Enott{e}) & \Myeq \emptyset \\
\Fst(\Eandd{e}) & \Myeq \emptyset
\end{align*}
\caption{Definition of $\Fst$ sets for regexes}
\label{fig:first}
\end{figure}

Figure~\ref{fig:first} gives an inductive definition for the $\Fst$
sets of extended regexes. For completeness, we also give cases for the
standard regexes. In our definition of $\Fst$ sets in terms of relation
$\Lr$ the $\Fst$ sets cannot include $\Epsi$, so expressions that
never consume any prefix of the subject have empty $\Fst$ sets. The
$\Fst$ sets of atomic groupings are conservative, as they may be a
proper superset of the first characters that the expression actually
consumes; for example, $\Fst(\Eindd{(\Chre{\Epsi}{a})})$ is $\{
\mathtt{a} \}$ instead of the more precise $\emptyset$.

\section{Conclusion}
\label{sec:conclusion}

We presented a new formalization of regular expressions that uses
natural semantics and a transformation $\Pi$ that converts a given
regular expression into an equivalent PEG, that is, a PEG that matches
the same strings that the regular expression matches in a
Perl-compatible regex implementation. If the regular expression's
language has the prefix property, easily guaranteed by using an
end-of-input marker, the transformation yields a PEG that recognizes
the same language as the regular expression.

We also have shown how our transformation can be easily adapted to
accommodate several ad-hoc extensions used by regex libraries:
independent expressions, possessive and lazy repetition, and
lookahead. Our transformation gives a precise semantics to what were
informal extensions with behavior specified in terms of how
backtracking-based regex matchers are implemented.

We show that, for some classes of regular expressions, we produce PEGs
that perform better by reasoning about how the PEG's limited
backtracking and syntactical predicates work to control the amount of
backtracking that a PEG will perform. The same reasoning that we apply
for large classes of expressions can be applied to specific ones to
yield bigger performance gains where necessary, although our
benchmarks show that simple optimizations are enough to perform close
to optimized regex matchers, while having a much simpler
implementation: both regex engines we used have over ten times the
amount of code of the PEG engine.

Another approach to establish the correspondence between regular
expressions and PEGs was suggested by
Ierusalimschy~\cite{roberto:lpeg}. In this approach we convert
Deterministic Finite Automata (DFA) into right-linear LL(1)
grammars. Medeiros~\cite{sergio:tese} proves that an LL(1) grammar has
the same language when interpreted as a CFG and as a PEG. But this
approach cannot be used with regex extensions, as they cannot be
expressed by a DFA.

The transformation $\Pi$ is a formalization of the continuation-based
conversion presented by Oikawa et al.~\cite{japa:sblp}. That work only
presents an informal discussion of the correctness of the conversion,
while we proved our transformation correct with regards to the
semantics of regular expressions and PEGs.

We can also benefit from the LPEG parsing
machine~\cite{dls:lpeg,roberto:lpeg}, a virtual machine for executing
PEGs. We can use the cost model of the parsing machine instructions to
estimate how efficient a given regular expression or regex is. The
parsing machine has a simple architecture with just nine basic
instructions and four registers, and implementations of our
transformation coupled with implementations of the parsing machine can
be the basis for simpler implementations of regex libraries.
 
\bibliographystyle{model3a-num-names}
\bibliography{repeg}

\end{document}